\documentclass[conference,10pt,letterpaper]{IEEEtran}

\usepackage{graphicx}
\usepackage{verbatim}
\usepackage{url}
\usepackage{amsmath}
\usepackage{amsfonts}
\usepackage{amssymb}
\usepackage{dsfont}
\usepackage[latin1]{inputenc}
\usepackage{delarray}
\usepackage{subfigure}
\usepackage{setspace}
\usepackage{paralist}
\usepackage{enumerate}

\IEEEoverridecommandlockouts

\newtheorem{theorem}{Theorem}[section]
\newtheorem{lemma}{Lemma}[section]
\newtheorem{remark}{Remark}[section]
\newtheorem{corollary}{Corollary}[section]
\newtheorem{example}{Example}[section]

\newcommand{\e}{{\rm e}}
\newcommand{\argmin}{\operatornamewithlimits{argmin}}

\newcommand{\maximize}{\operatornamewithlimits{maximize~~~}}
\newcommand{\minimize}{\operatornamewithlimits{minimize~~~}}
\newcommand{\subjectto}{\operatornamewithlimits{subject~to~~~}}

\newcommand{\remove}[1]{}

\remove{

\newtheorem{theorem}{Theorem}[section]

}

\DeclareMathSizes{10}{9}{7}{5}

\begin{document}
\title{Adaptive Spatial Aloha, Fairness and Stochastic Geometry}
\author{
\IEEEauthorblockN{Francois Baccelli}
\IEEEauthorblockA{
  Depts. of Mathematics and ECE\\
  University of Texas at Austin, USA\\
  and \\
  INRIA Rocquencourt Paris, France \\
  Email: francois.baccelli@austin.utexas.edu}
  \and
 \IEEEauthorblockN{Chandramani Singh}
 \IEEEauthorblockA{
  INRIA-TREC\\
  23 Avenue d'Italie\\
  CS 81321\\
  75214 Paris Cedex 13, France\\
  Email: chandramani.singh@inria.fr}
  }

\maketitle

\begin{abstract}
This work aims at combining adaptive protocol design,
utility maximization and
stochastic geometry. We focus on a spatial adaptation of
Aloha within the framework of ad hoc networks. We consider
quasi-static networks in which mobiles learn the local topology
and incorporate this information to adapt their medium access
probability~(MAP) selection to their local environment. We
consider the cases where nodes cooperate in a distributed way
to maximize the global throughput or to achieve either
proportional fair or max-min fair medium access.
In the proportional fair case, we show that nodes can compute
their optimal MAPs as solutions to certain fixed point equations.
In the maximum throughput case, the optimal MAPs are obtained
through a Gibbs Sampling based algorithm.
In the max min case, these are obtained
as the solution of a convex optimization problem.
The main performance analysis result of the paper
is that this type of distributed adaptation can be analyzed
using stochastic geometry in the proportional fair case.
In this case, we show that, when the nodes form a homogeneous
Poisson point process in the Euclidean plane,
the distribution of the optimal MAP
can be obtained from that of a certain shot noise process
w.r.t. the node Poisson point process
and that the mean utility can also be derived from this distribution.
We discuss the difficulties to be faced for
analyzing the performance of the other cases (maximal
throughput and max-min fairness).
Numerical results illustrate our findings and
quantify the gains brought by spatial adaptation in
such networks.
\end{abstract}

\remove{




}

\section{Introduction}
Stochastic geometry has recently been used for the analysis
and performance evaluation of wireless~(ad hoc as well as cellular) networks;
in this approach, one models node locations as a spatial point process,
e.g., homogeneous Poisson point processes, and one
computes various network statistics, e.g.,  interference,
successful transmission probability, coverage~(or, outage)
probability etc. as spatial averages. This often leads
to tractable performance metrics that are amenable to parametric optimization
with respect to network parameters (node density, protocol parameters, etc.).
More precisely, this approach yields spatial averages
of the performance metrics for given network parameters;
then the parameters can be chosen to optimize performance.
This approach takes a macroscopic view of the network with
the underlying assumption that all nodes in the network have
identical statistical characteristics.

In practice, due to randomness and heterogeneity in networks,
nodes need to adapt to local spatial and temporal conditions~(e.g.,
channel conditions and topology) to reach optimum network wide
performance. For example, nodes in wireless LANs adjust their
window sizes based on acknowledgment feedback; in cellular networks
nodes are scheduled based on channel conditions
and adapt their transmit powers based on the measured
SINRs, which in turn depend on the transmit powers set by other nodes.
In all such scenarios, distributed adaptive algorithms are used
to reach a desired network wide operating point e.g. that maximizing
some utility.
While the behavior of such
distributed optimization protocols is often well understood on a given
topology, there are usually no analytical characterizations of the
statistical properties of the optimal state in large
random and heterogeneous networks.

The main aim of this work is to use
stochastic geometry to study spatial adaptations
of medium access control in Aloha that aim at optimizing
certain utilities.
While we identify a
utility for which stochastic geometry can be used to compute the
spatial distribution of MAP and the expected utility, we are far
from being able to do so for all types of utilities within
the $\alpha$-fair class and we discuss the difficulties
to be faced.

Let us start with a review of the
state of the art on Aloha.
Wireless spectrum is well known to be
a precious and scarce shared resource.
Medium Access Control~(MAC) algorithms are employed
to coordinate access to the shared
wireless medium. An efficient MAC protocol should ensure
high system throughput, and should also distribute the
available bandwidth fairly among
the competing nodes. The simplest of the MAC protocols,
Aloha and slotted Aloha, with a "random access" spirit,
were introduced by Abramson~\cite{ctrltheory-comnets.abramson70aloha}
and Roberts~\cite{ctrltheory-comsnets.roberts75aloha-slots-capture}
respectively. In these protocols, only one node could successfully
transmit at a time. Reference~\cite{ctrltheory-wireless.baccelli-etal06aloha-multihop-wireless}
modeled node locations as spatial point processes, and also modeled
channel fadings, interferences and SINR based reception.
This allowed for spatial reuse and multiple simultaneous successful
transmissions depending on SINR levels at the
corresponding receivers. All the above protocols prescribe
identical attempt probabilities for all the nodes.
Reference~\cite{ctrltheory-wireless.baccelli-etal09spatial-opportunistic-aloha}
further proposed opportunistic Aloha in which
nodes' transmission attempts are modulated by their
channel conditions.

Among the initial attempts of MAP adaptation in Aloha, reference~\cite{ctrltheory-wireless.hazek85-decentralized-control-multiaccess}
analyzed protocol model and proposed stochastic approximation based strategies
that were based on receiver feedback and were aimed at stabilizing the network.
References~\cite{ctrltheory-wireless.baccelli-etal06aloha-multihop-wireless,
ctrltheory-wireless.baccelli-etal09spatial-opportunistic-aloha} also
optimized nodes'
attempt probabilities~(or thresholds) in order to maximize the
spatial density of successful transmissions.
Reference~\cite{ctrltheory-wireless.hsu-su11channel-aware-aloha} analyzed
both plain and opportunistic Aloha
in a network where all the nodes communicate to one access point.
They assumed statistically identical Rayleigh faded channels
with no dependence on geometry~(i.e., no path loss components).
They demonstrated a paradoxical behavior where
plain Aloha yields better aggregate throughput than the opportunistic one.
Reference~\cite{ctrltheory-wireless.mohsenian-rad-et-al10sinr-random-access} also
studied optimal random access with SINR based reception. However,
they considered constant channel gains. They developed a centralized
algorithm that maximizes the network throughput, and also an algorithm
that leads to max-min fair operation.
Reference~\cite{ctrltheory-wireless.wang-kar04max-min-fair-aloha}
modeled network as an undirected graph and studied Aloha under
the protocol model.
They designed distributed algorithms that are
either proportional fair or max-min fair.
Reference~\cite{gametheory-wireless.hanawal-etal12medium-access-games}
built upon the model
of~\cite{ctrltheory-wireless.baccelli-etal06aloha-multihop-wireless},
and formulated the channel access problem as a non-cooperative game
among users.
They considered throughput and delay as performance metrics and proposed
pricing schemes that induce socially optimum behavior at equilibrium.
However, they set time average quantities~(e.g., throughput, delay)
as utilities~(or costs),
and concentrated on symmetric Nash equilibria.
Consequently, in their analysis,
dependence on local conditions vanishes.


In none of the above Aloha models, nodes account for both wireless channel
randomness and local topology for making their random access decisions,
as we do in the present paper.\footnote{In view of this distinction, we refer to the spatial
Aloha protocol of~\cite{ctrltheory-wireless.baccelli-etal06aloha-multihop-wireless} as
plain Aloha.}

There is a vast literature on the modeling
of CSMA by stochastic geometry which will not be
reviewed in detail here.
The very nature of this MAC protocol is adaptive as each node senses
the network and acts in order to ensure
that certain exclusion rules are satisfied,
namely that neighboring nodes do not
access the channel simultaneously.
However, CSMA as such is designed to guarantee a reasonable scheduling,
not to optimize any utility of the throughput.
The closest reference to our work is probably~\cite{DBLP:journals/questa/BaccelliLRSSW12}
where the authors study an adaptation of the
exclusion range and of the transmit power
of a CSMA node to the location of the closest
interferer. This adaptation aims at
maximizing the mean number
of nodes transmitting per unit time and space~(while respecting the above exclusion rules).
This mean number is however only a surrogate of the rate.
In addition, the adaptation is only w.r.t. to the location of the nearest interferer.

We study spatial adaptation of Aloha in ad hoc networks.
The network setting is described in Section~\ref{sec:network-model}.
We consider quasi-static networks in which mobiles
learn the topology, and incorporate this information
in their medium access probability~(MAP) selection.
We consider the cases where nodes are benevolent and cooperate
in a distributed way to maximize the
global network throughput or to reach either a proportional fair or
max-min fair sharing of the network resources.
We analyze the case where nodes account only for their closest interferers,
for all nodes in a given ball around them, or even all nodes in the network.

Section~\ref{sec:aggr-interf} is focused on the distributed
algorithms that maximize the aggregate throughput or
lead to max-min fairness in such networks.
In the proportional fair case,
we show that nodes can compute the optimal
MAPs as solutions to certain
fixed point equations.
In the maximum throughput case, the optimal MAPs are obtained
through a Gibbs Sampling based algorithm.
In the max min case, the optimal MAPs are obtained
as the solution of a convex optimization problem.

Section~\ref{ss:stocg} contains the stochastic geometry results.
The model features nodes forming a realization of a
homogeneous Poisson point process in the Euclidean plane.
We compute the MAP distribution in such a network
in the proportional fair case using shot noise field theory.
To the best of our knowledge, this distribution is the first
example of successful combination of stochastic geometry and
adaptive protocol design aimed at optimizing ceratin utility function
within this Aloha setting. We also show that
the mean value of the logarithm of the throughput
obtained by a typical node can be derived from this distribution.
Finally, we discuss the difficulties to be faced in order
to extend the result to other types of utilities.

The numerical results are gathered in Section~\ref{ss:numerics}.
The aim of this section is two-fold: 1) check the analytical
results against simulation and 2) quantify the gains brought
by adaption within this setting.

\section{Network Model}
\label{sec:network-model}

\begin{figure}[h]

\includegraphics[width=3.0in,height=2.5in]{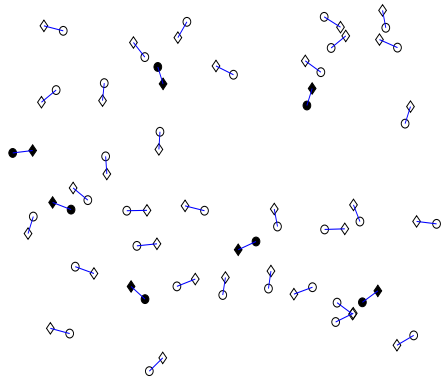}
\caption{A snapshot of bipolar MANET with Aloha as the medium access protocol.
The {\it diamonds} represent transmitters, and the connected {\it circles} the corresponding receivers.
The solid diamonds represent the nodes that are transmitting in a slot.}
\label{fig:bipole-model}
\end{figure}

We model the ad-hoc wireless network as a set of transmitters and
their corresponding receivers, all located
in the Euclidean plane. This is often referred to as
``bipole model''~\cite[Chapter~16]{stochproc-wireless.baccelli-blaszczyszyn09stochastic-geometry-wireless-networks-2}.
There are $N$ transmitter-receiver pairs communicating over a shared channel.
The transmitters follow the slotted version of the
Aloha medium access control~(MAC) protocol~(see Figure~\ref{fig:bipole-model}).
A transmitter, in each transmission attempt,
sends one packet which occupies one slot.
Each transmitter uses unit transmission power.
We assume that each node has an infinite backlog
of packets to transmit to its receiver.
The Euclidean distance between transmitter $j$ and receiver
$i$ is $r_{ji}$, and the path-loss exponent is $\alpha$~($\alpha > 2$).
We also assume Rayleigh faded channels with
$h_{ji}$ being the random fading between
transmitter $j$ and receiver
$i$. Moreover, we assume that the random variables $h_{ij}, 1 \leq i \leq N, 1 \leq j \leq N$
are independent and identically distributed with mean $1/\mu$.\footnote{The independence
assumption is justified if the distance between two receivers is larger than
the coherence distance of the wireless
channel~\cite{stochproc-wireless.baccelli-blaszczyszyn09stochastic-geometry-wireless-networks-2}.
We assume this to be the case.} Thus all $h_{ji}$s have cumulative distribution function~(cdf)
$F(x) = 1 - e^{-\mu x}$ with $x \geq 0$.
\remove{
The $i$th transmitter, whenever it accesses the
medium, transmits with power $p_i$.
We assume no power control, and instead, adapt medium access probabilities
based on the perceived interference.
The power gain between $i$th transmitter and $j$th receiver
is $h_{ij}$, and is assume to be constant over the scale of interest.}
All the receivers are also subjected to white Gaussian thermal noise
with variance $w$, which is also constant across slots.
We assume that a receiver successfully receives the packet of the corresponding transmitter
if the received SINR exceeds a threshold $T$.

Let $e_i$ be the indicator variable indicating whether transmitter $i$
transmits in a slot, and $p_i$ be $i$'s medium access probability.
Thus $\mathbb{P}(e_i = 1) = p_i.$
When node $i$ transmits, the received SINR at the corresponding receiver
is
\[
\gamma_i = \frac{h_{ii} r_{ii}^{-\alpha}}{\sum_{j \neq i} e_j h_{ji} r_{ji}^{-\alpha}  + w}.
\]
Then the probability of successful reception $q_i$ can be calculated as follows.
\begin{align*}
\lefteqn{\mathbb{P}\left(\gamma_i \geq T | \{(h_{ji},e_j): j \neq i\} \right)} \\
& = \mathbb{P} \left(h_{ii} \geq \sum_{j \neq i} e_j h_{ji} \left(\frac{r_{ji}}{r_{ii}}\right)^{-\alpha} T + \frac{wT}{r_{ii}^{-\alpha}} \bigg| \{(h_{ji},e_j): j \neq i\} \right) \\
& = \exp \left(-\mu T \left(\sum_{j \neq i} e_j h_{ji} \left(\frac{r_{ji}}{r_{ii}}\right)^{-\alpha} + \frac{w}{r_{ii}^{-\alpha}}  \right)\right).
\end{align*}
Thus
\begin{align*}
q_i & = \mathbb{E}_{\{(h_{ji},e_j): j \neq i\}}\exp \left(-\mu T \left(\sum_{j \neq i} e_j h_{ji} \left(\frac{r_{ji}}{r_{ii}}\right)^{-\alpha} + \frac{w}{r_{ii}^{-\alpha}}  \right)\right) \\
& = \e^{\frac{-\mu wT}{r_{ii}^{-\alpha}}}\prod_{j \neq i} \mathbb{E}_{(h_{ji},e_j)}\exp \left(-\mu e_j h_{ji} \left(\frac{r_{ji}}{r_{ii}}\right)^{-\alpha} T \right) \\
& = \e^{\frac{-\mu wT}{r_{ii}^{-\alpha}}}\prod_{j \neq i} \mathbb{E}_{h_{ji}}\left((1 - p_j) + p_j \exp \left(-\mu h_{ji} \left(\frac{r_{ji}}{r_{ii}}\right)^{-\alpha} T \right) \right) \\
& = \e^{-\mu w T r_{ii}^{\alpha}} \prod_{j \neq i} \left((1 - p_j) + \frac{p_j}{1 + 1 / b_{ji}} \right),
\end{align*}
where
$
b_{ji} = \frac{1}{T} \left(\frac{r_{ji}}{r_{ii}}\right)^{\alpha}.
$
Further simplifying,
\begin{equation}
\label{eqn:succ-prob}
q_i = \e^{-\mu w T r_{ii}^{\alpha}} \prod_{j \neq i}\left(1 - \frac{p_j}{1 + b_{ji}}\right).
\end{equation}
Then, the rate or throughput of transmitter $i$ is given by $p_i q_i$.

The thermal noise appears merely as a constant multiplicative factor in the
expression for the successful transmission probability~(see~\eqref{eqn:succ-prob}).
Moreover, in interference limited networks, the impact of thermal
noise is negligible as compared to interference.
We focus on such networks, and thus we ignore the thermal noise factor throughout.

\section{Adaptive Spatial Aloha and Fairness}
\label{sec:aggr-interf}
In this section, we analyze adaptations of spatial Aloha that maximize
aggregate throughput or achieve proportional fairness or max-min fairness.

\subsection{Maximum Throughput Medium Access}
\label{sec:aggr-interf-max-thput}
The throughput maximizing medium access probabilities solve the following
optimization problem.
\begin{align*}
\maximize  & \Theta := \sum_i p_i \prod_{j \neq i}\left(1 - \frac{p_j}{1 + b_{ji}}\right) , \\
\subjectto & 0 \leq p_i \leq 1, \ i \in \mathcal{N}.
\end{align*}

We first argue that the optimum in the above optimization problem
is attained at one of the vertices of the hypercube formed by the
constraint set. To see this, suppose $\mathbf{p^{\ast}} \in [0,1]^{\mathcal{N}}$
is an optimal solution, and $p^{\ast}_i \in (0,1)$ for some $i \in \mathcal{N}$.
Clearly,
\begin{align*}
\lefteqn{\frac{\partial \Theta}{\partial p_i} \Big|_{\mathbf{p} = \mathbf{p^{\ast}}}} \\
 & = \prod_{j \neq i}\left(1 - \frac{p^{\ast}_j}{1 + b_{ji}}\right) - \sum_{j \neq i} \frac{p^{\ast}_j}{1 + b_{ij}} \prod_{k \neq i,j} \left(1 - \frac{p^{\ast}_k}{1 + b_{kj}}\right) \\
 & = 0.
\end{align*}
Since the partial derivative is independent of $p_i$, $p_i$ can be set to either $0$ or $1$
without reducing the value of the objective function. This proves our claim.
In the following we focus only on such
extreme solutions. Then the above problem is
equivalent to finding an $\mathcal{M} \subset{N}$ such that
$p_i = 1$ if and only if $i \in \mathcal{M}$ is an optimal solution.
Thus we are interested in
\[
\maximize_{\mathcal{M} \subset \mathcal{N}} \sum_{i \in \mathcal{M}} \prod_{j \in \mathcal{M} \setminus \{i\}} \left(1 -\frac{1}{1 + b_{ji}} \right).
\]

\paragraph*{An iterative solution}
We can pose this problem as a strategic form game with the users as
players~\cite{gametheory.osborne-rubinstein99game-theory}. For
each player its action $a_i$ lies in $\{0,1\}$, and the utility function
$u_i: \mathbf{a} \mapsto \mathbb{R}$
is given by
\begin{align*}
u_i(0,\mathbf{a}_{-i}) = 0&, \\
u_i(1,\mathbf{a}_{-i}) = 1& - \prod_{j \in \mathcal{M} \setminus \{i\}}\left(1 - \frac{1}{1 + b_{ji}}\right)\\
                         & - \sum_{j \in \mathcal{M} \setminus \{i\}} \frac{1}{1 + b_{ij}} \prod_{k \in \mathcal{M} \setminus \{i,j\}} \left(1 - \frac{1}{1 + b_{kj}}\right).
\end{align*}
This is a potential game with the above objective function as the potential
function~\cite{gametheory.monderer-shapley96potential-games}.
Thus the best response dynamics converges to a Nash equilibrium.
This algorithm can be implemented in a distributed fashion if each node $i$ knows
$b_{ij}, b_{ji}$ for all $j$, and also $\mathcal{M}$ and
$\prod_{k \in \mathcal{M} \setminus \{j\}}(1 - (1 +  b_{kj})^{-1})$ for all $j \in \mathcal{M}$ after each
iteration. However, a Nash equilibrium can be a
suboptimal solution to the above optimization problem. To alleviate this problem, we propose a
Gibbs sampler based distributed algorithm, wherein each node $i$ chooses
action $1$ with probability
\[
p_i = \frac{\e^{u_i(1,\mathbf{a}_{-i})/{\tau}}}{1 + \e^{u_i(1,\mathbf{a}_{-i})/{\tau}}}.
\]
The parameter $\tau$ is called the temperature. The Gibbs
sampler dynamics converges to a steady state which is the
Gibbs distribution associated with the aggregate throughput and
the temperature $\tau$~\cite{stochproc.bremaud99markov-chains}. In other words, we are led to the following
distribution on the action profiles:
\[\pi_{\tau}(\mathbf{a}) =  u \e^{\sum_{i \in \mathcal{N}}u_i(\mathbf{a})},\]
where $u$ is a normalizing constant.
When $\tau$ goes to $0$ in an appropriate way~(i.e.,
as $1/\log(1 + t)$, where $t$ is the time), the distribution $\pi_{\tau}(\cdot)$
converges to a dirac mass at the action profile $\mathbf{a^{\ast}}$ with maximum
aggregate utility if it is unique. Notice that the aggregate utility $\sum_{i \in \mathcal{N}}u_i(\mathbf{a})$
equals the aggregate throughput. Thus the action profile $\mathbf{a^{\ast}}$ is a solution
to the original throughput optimization problem.
\begin{remark}
The first two terms in the utility function $u_i(1,\mathbf{a}_{-i})$
can be seen as ``selfish'' part of user $i$, whereas the last
summation term is ``altruistic'' part. The user makes a decision based
on whether the ``selfish'' part dominates or viceversa.
 \end{remark}

 \begin{remark}
 In a quasi-static network where topology continuously changes, although at a slower time scale,
 different sets of nodes are likely to be scheduled to transmit under different topologies.
 Thus, in terms of long term performance, maximum throughput medium access is not
 grossly unfair.
 \end{remark}

\subsection{Proportional Fair Medium Access}
\label{sec:aggr-interf-prop-fair}
The proportional fair medium access problem can be formulated as follows.
\begin{align*}
\maximize & \sum_i \log{(p_i q_ i)}, \\
\subjectto & 0 \leq p_i \leq 1, \ i \in \mathcal{N}.
\end{align*}
The objective function can be rewritten as
\[
\sum_i \left( \log{p_i} + \sum_{j \neq i}\log \left(1 - \frac{p_j}{1 + b_{ji}}\right) \right).
\]
We thus have a convex separable optimization problem. The partial derivative
of the objective function with respect to $p_i$ is
\begin{equation}
\label{pfair-fixed-point}
\frac{1}{p_i} - \sum_{j \neq i}\frac{1}{1 + b_{ij} - p_i},
\end{equation}
which is continuous and decreasing in $p_i$ over $[0,1]$.
We conclude that at optimality, for each $1 \leq i \leq N$,
\begin{equation*}
p_i = \begin{cases}
      f_i(p_i) :=  \left(\sum_{j \neq i}\frac{1}{1 + b_{ij} - p_i} \right)^{-1} & \mbox{if } f_i(1) < 1, \\
      1 & \mbox{otherwise.}
      \end{cases}
\end{equation*}
Observe that user $i$'s optimal attempt probability is independent of the
attempt probabilities of other users. In particular, if $f_i(1) < 1$, user $i$ can
perform iterations $p_i^{k+1} = f_i(p_i^k)$ autonomously.
Furthermore,
\[
f_i'(p_i) = - \sum_{j \neq i}\frac{1}{(1 + b_{ij} - p_i)^2} \left(\sum_{j \neq i}\frac{1}{1 + b_{ij} - p_i} \right)^{-2}.
\]
Clearly, $|f_i'(p_i)| < 1$, i.e., $f_i(\cdot)$ is a contraction. Thus the fixed point iterations
converge to the optimal $p_i$ starting from any $p_i \in [0,1]$.
\begin{remark}
The characterization of the optimal attempt probabilities reflects the altruistic behavior
of users. More precisely, user $i$'s attempt probability is a function of
$\{b_{ij}, j \neq i\}$ which are measures of $i$'s interference to all other users.
In particular, if $\sum_{j \neq i}\frac{1}{b_{ij}} < 1$, i.e., if $i$'s transmission
does not cause significant interference to the other users, then
$i$ transmits in all the slots. Unlike the throughput maximization problem, there is no
``selfish'' component in the decision making rule.
\end{remark}

\begin{remark}
As the target SINR $T \to \infty$, $b_{ij} \to 0$ for all $i,j$,
and the proportional fair attempt probabilities satisfy
\begin{align*}
\frac{1}{p_i} = \sum_{j \neq i}\frac{1}{1 - p_i} = \frac{N-1}{1- p_i}
\end{align*}
for all $i \in \mathcal{N}$. This yields $p_i = \frac{1}{N}$ for all $i \in \mathcal{N}$.
This is expected, because in the limiting case
a transmission can succeed if and only if there is no other
concurrent transmission. This is hence
Aloha without spatial reuse, and it is well known
that in this case, the optimal access probability is $1/N$
asymptotically~\cite{commnet.bertsekas-gallager92data-networks}.
 \end{remark}

\subsection{Max-min Fair Medium Access}
\label{sec:aggr-interf-max-min-fair}
Our analysis in this section
follows~\cite{ctrltheory-wireless.wang-kar04max-min-fair-aloha,ctrltheory-wireless.wang-etal06lexicographic-max-min-fair-access}. The max-min fair medium access problem can be formulated as
\begin{align*}
\maximize & \theta, \\
\subjectto & \theta \leq p_i \prod_{j \neq i}\left(1 - \frac{p_j}{1 + b_{ji}}\right), \ i \in \mathcal{N},
\end{align*}
where constraint functions are defined for all $\mathbf{p} \in [0,1]^{\mathcal{N}}$.
The following is an equivalent convex optimization problem~(see~\cite{ctrltheory-wireless.wang-kar04max-min-fair-aloha}
for details):
\begin{align*}
\minimize & \frac{1}{2}\theta^2, \\
\subjectto & \theta \leq \log p_i + \sum_{j \neq i}\log \left(1 - \frac{p_j}{1 + b_{ji}}\right), \ i \in \mathcal{N}.
\end{align*}

The Lagrange function of this problem is given by~\cite{ctrltheory.boyd-vandenberghe04convex-optimization}
\[
\frac{1}{2}\theta^2 + \sum_{i \in \mathcal{N}} \lambda_i\left(\theta - \log p_i - \sum_{j \neq i}\log \left(1 - \frac{p_j}{1 + b_{ji}}\right) \right),
\]
with $\lambda_i \geq 0, i \in \mathcal{N}$ being the Lagrange multipliers.

Minimization of the Lagrange function~(which is concave in $\mathbf{p}$ and $\theta$) gives
\begin{align}
p_i = &\begin{cases}
       \lambda_i \left(\sum_{j \neq i}\frac{\lambda_j}{1 + b_{ij} - p_i} \right)^{-1} & \mbox{if } \frac{1}{\lambda_i}\sum_{j \neq i}\frac{\lambda_j}{b_{ij}} > 1, \\
       1 & \mbox{otherwise.}
      \end{cases}
\label{eqn:aggr-interf-p-update} \\
\theta = &-\sum_{i \in N}\lambda_i.
\label{eqn:aggr-interf-theta-update}
\end{align}
 Wang and Kar~\cite{ctrltheory-wireless.wang-kar04max-min-fair-aloha} suggest that the Lagrange multipliers be updated using the gradient projection method. More precisely, for all $i \in \mathcal{N}$,
\begin{align}
\lambda_i(n+1) = &\Bigg[\lambda_i(n) + \beta(n) \\
                 &\left. \left(\theta - \log p_i - \sum_{j \neq i}\log \left(1 - \frac{p_j}{1 + b_{ji}}\right) \right)\right]^+, \label{eqn:aggr-interf-lambda-update}
\end{align}
where $\beta(n)$ is the step size at the $n$th iteration.
Further more,~\cite[Theorem~2]{ctrltheory-wireless.wang-kar04max-min-fair-aloha}
implies that a solution arbitrary close to an optimal solution can be reached via
appropriate choice of step sizes. However, all the users need to exchange variables in order to perform updates.

Finally, the {\it directed link graph} corresponding to our network is a directed graph in which
each vertex stands for a user~(i.e., a transmitter-receiver pair) in the network. There
is an edge from vertex $i$ to vertex $j$ in the directed link graph if transmission of user $i$
affects the success of transmission of user $j$.
Two vertices $i$ and $j$ are said to be connected if either of the following two
conditions hold:
\begin{enumerate}
\item there is an edge from $i$ to $j$ or viceversa,
\item there are vertices $v_0 = i, v_1,\dots,v_{n-1},v_n = j$ such that
      $v_m$ and $v_{m+1}$ are connected for $m = 0,1,\dots,n-1$.
\end{enumerate}
Clearly, the directed link graph for our network model
is a complete graph; for any pair of vertices
$i$ and $j$ there is an edge from $i$ to $j$ and also from $j$ to $i$. In particular,
the directed link graph is a single {\it strongly connected component}~\cite{ctrltheory-wireless.wang-etal06lexicographic-max-min-fair-access}.
Thus~\cite[Corollary~1]{ctrltheory-wireless.wang-etal06lexicographic-max-min-fair-access}
implies that the above optimization also obtains the lexicographic max-min fair
medium access probabilities that yield identical rates for all the users.

\subsection{Closest Interferer Case}
\label{sec:close-interf}
Note that a user needs to know the entire topology, and in a few cases, also needs to
communicate with all the nodes to implement the adaptation rules
developed in Sections~\ref{sec:aggr-interf-max-thput}-\ref{sec:aggr-interf-max-min-fair}. In this section,
we carry out analysis under the simplifying assumption
that the aggregate interference
at a receiver is dominated by the transmission from the closest interferer.
This is a reasonable approximation in a moderately dense network, specifically
when the path loss attenuations are high. Throughout this section, we use the
notation
\begin{align*}
& c(i) := \argmin_{j \neq i}r_{ji}, \\
& C(i) := \{j: c(j) = i\}, \\
\end{align*}
for all $1 \leq i \leq N$. In words, $c(i)$ is the strongest interferer of node $i$,
and $C(i)$ is the set of nodes to which node $i$ is the strongest interferer.
We assume that there is always a unique  $c(i)$ for each $i$.
Then, accounting only for the nearest interferer,
the approximate probability of successful transmission for node $i$ is
\[\tilde{q}_i = 1 - \frac{p_{c(i)}}{1 + b_{c(i)i}}.\]
The analysis of Sections~\ref{sec:aggr-interf} can be adapted
to this simplified scenario.

\subsubsection{Maximum Throughput Medium Access}
\label{sec:close-interf-max-thput}
The throughput maximization problem can now be posed as follows.
\begin{align*}
\maximize  & \tilde{\Theta} := \sum_i p_i \tilde{q}_ i , \\
\subjectto & 0 \leq p_i \leq N, \ i \in \mathcal{N}.
\end{align*}

As in Section~\ref{sec:aggr-interf-max-thput}, we can argue that some $\mathbf{p}^{\ast} \in \{0,1\}^{\mathcal{N}}$
attains the optimal throughput. Again, an equivalent optimization problem is
\[
\maximize_{\mathcal{M} \subset \mathcal{N}}\sum_{i \in \mathcal{M}} \left(1 - \frac{\mathds{1}\{c(i) \in \mathcal{M}\}}{1 + b_{c(i)i}} \right),
\]
 or alternatively,
 \[
\maximize_{\mathcal{M} \subset \mathcal{N}}\sum_{i \in \mathcal{M}} \left(1 - \sum_{j \in C(i)} \frac{\mathds{1}\{j \in \mathcal{M}\}}{1 + b_{ij}} \right).
\]

We now formulate a strategic form game among users, with action sets $\{0,1\}$ and
utility functions given by
\begin{align*}
u_i(0,\mathbf{a}_{-i}) = & 0, \\
u_i(1,\mathbf{a}_{-i}) = & 1 - \frac{a_{c(i)}}{1 + b_{c(i)i}} - \sum_{j \in C(i)} \frac{a_j}{1 + b_{ij}}.
\end{align*}
Again a Gibbs sampler based algorithm yields the optimal set of the transmitting users.
Also, user $i$ only needs to know the distances of user $c(i)$ and all the receivers in $C(i)$
and their actions to make its decision.

\begin{remark}
Notice that user $i$ must choose $a_i =1$ if
\[
1 - \frac{1}{1 + b_{c(i)i}} - \sum_{j \in C(i)} \frac{1}{1 + b_{ij}} > 0.
\]
Such users can set their actions to $1$, and need not undergo Gibbs sampler
updates.
\end{remark}
\paragraph*{Discussion}
Consider a scenario where a node's closest interferer does not transmit, i.e., has
zero attempt probability. Nonetheless, this node always has an active closest
interferer~(unless there are no other nodes in the network). A better
approximation of the success probabilities, and hence of the throughput, is
obtained by always accounting for the closest active interferer. Towards this,
let us define
\begin{align*}
& c(i,\mathcal{M}) := \argmin_{j \in \mathcal{M}, j \neq i}r_{ji}, \\
& C(i,\mathcal{M}) := \{j \in \mathcal{M}: c(j,\mathcal{M}) = i\}, \\
\end{align*}
for all $i \in \mathcal{M}$.
We are now faced with the following optimization problem.
\[
\maximize_{\mathcal{M} \subset \mathcal{N}}\sum_{i \in \mathcal{M}} \left(1 - \sum_{j \in C(i,\mathcal{M})} \frac{1}{1 + b_{ij}} \right).
\]
 We can now define  users' utility functions as follows.
\begin{align*}
u_i(0,\mathbf{a}_{-i}) =&  0, \\
u_i(1,\mathbf{a}_{-i}) =& 1 - \frac{1}{1 + b_{c(i,\mathcal{M})i}} \\
                        & - \sum_{j \in C(i,\mathcal{M} \cup \{i\})} \left( \frac{1}{1 + b_{ij}} - \frac{1}{1 + b_{c(j,\mathcal{M})j}}\right),
\end{align*}
where $\mathcal{M} = \{j \in \mathcal{N}: j \neq i, a_j = 1\}$.
The analogous distributed algorithm~(Gibbs sampler based) can again be shown to lead to
the optimal solution.

\subsubsection{Proportional Fair Medium Access}
\label{sec:close-interf-prop-fair}
We now aim to solve the following optimization problem.
\begin{align*}
\maximize \sum_i \log{(p_i \tilde{q}_ i)}, \\
\subjectto 0 \leq p_i \leq 1, \ i \in \mathcal{N}.
\end{align*}
Following the discussion in Section~\ref{sec:aggr-interf-prop-fair}, we obtain
\begin{equation*}
p_i = \begin{cases}
      \left(\sum_{j \in C(i)}\frac{1}{1 + b_{ij} - p_i} \right)^{-1} & \mbox{if } \sum_{j \in C(i)}\frac{1}{b_{ij}} > 1, \\
      1 & \mbox{otherwise.}
      \end{cases}
\end{equation*}
Again, if $\sum_{j \in C(i)}\frac{1}{b_{ij}} > 1$, iterations $p_i^{k+1} = f_i(p_i^k)$
converge to the optimal $p_i$ starting from any $p_i \in [0,1]$.

\begin{remark}
If $\sum_{j \in C(i)}\frac{1}{b_{ij}} < 1$, i.e., if $i$'s transmission
does not cause significant interference to the users for whom $i$ is closest interferer, then
$i$ transmits in all the slots.
The same user may not transmit~(in any slot) under the maximum throughput objective if
$1 - \frac{1}{1 + b_{c(i)i}} \approx 0$ and $c(i)$ and users in $C(i)$ transmit.
\end{remark}

We can have explicit formulae for the attempt probabilities in a few
special cases.
\begin{enumerate}
\item Suppose $C(i)$ is singleton for each $i$. If $C(i) = \{j\}$, then
\begin{equation*}
p_i = \begin{cases}
       \frac{1 + b_{ij}}{2} & \mbox{if } b_{ij} < 1, \\
      1 & \mbox{otherwise.}
      \end{cases}
\end{equation*}

\item {\it Linear Network Topology:~}
We now consider a scenario where nodes are placed along a line, say $\mathbb{R}$, and
are indexed sequentially. We also assume that for any node $i$ the potential interferers are
the two immediate neighbors $i-1$ and $i+1$. This also amounts to assuming $C(i) = \{i-1, i+1\}$
for all $i$. Then, assuming $\frac{1}{b_{i,i-1}} + \frac{1}{b_{i,i+1}} > 1$, the proportional fair attempt probability of node $i$ satisfies
\[
\frac{1}{p_i} = \frac{1}{1+b_{i,i-1}-p_i} + \frac{1}{1+b_{i,i+1}-p_i}.
\]
This is quadratic equation in $p_i$, which on solving gives\footnote{The other root is greater than 1, and thus
is not a valid probability.}
\begin{align*}
p_i
    &= \frac{2 + b_{i,i-1} + b_{i,i+1} - \sqrt{(b_{i,i-1} - b_{i,i+1})^2 \atop + (1 + b_{i,i-1})(1 + b_{i,i+1})}}{3}.
\end{align*}
\end{enumerate}

\subsubsection{Max-min Fair Medium Access}
\label{sec:close-interf-max-min-fair}
Similar to Section~\ref{sec:aggr-interf-max-min-fair}, the max-min fair medium access problem is
\begin{align*}
\maximize & \theta, \\
\subjectto & \theta \leq p_i \left(1 - \frac{p_{c(i)}}{1 + b_{c(i)i}}\right), \ i \in \mathcal{N}.
\end{align*}
Again, the constraint functions are defined for all $\mathbf{p} \in [0,1]^{\mathcal{N}}$.

Let us recall the definition of the directed link graph associated
with the network. In this section, we only account
for the interference due the closest interferer. Thus there is an
edge from vertex $i$ to vertex $j$ if and only if $j \in C(i)$.
We assume that the directed link graph is connected~(i.e., all the vertices in the graph are
connected with each other). If it is not connected, the max-min fair medium access problem
on the entire graph decomposes into separate max-min
fair medium access problems on each of the connected subgraphs, which can be solved independently.

We now pursue the following convex optimization problem which is equivalent to
the above max-min fair medium access optimization problem.
\begin{align*}
\minimize & \frac{1}{2}\sum_{i \in \mathcal{N}}\theta_i^2, \\
\subjectto & \theta_i \leq \log p_i + \log \left(1 - \frac{p_{c(i)}}{1 + b_{c(i)i}}\right), \ i \in \mathcal{N}, \\
           & \theta_i \leq \theta_j, j \in C(i), j = c(i).
\end{align*}
The last set of constraints along with the connected assumption~(of the directed link graph)
forces $\theta_i$ to be equal for all $i \in \mathcal{N}$. This confirms equivalence to the
initial optimization problem.
Now the Lagrange function is
\begin{align*}
\frac{1}{2}\sum_{i \in \mathcal{N}}\theta_i^2 &+ \sum_{i \in \mathcal{N}} \lambda_i\left(\theta_i - \log p_i - \log \left(1 - \frac{p_{c(i)}}{1 + b_{c(i)i}}\right) \right) \\
                      &+ \sum_{i \in \mathcal{N}} \sum_{j \in C(i) \cup \{c(i)\}} \mu_{ij}(\theta_i - \theta_j),
\end{align*}
where $\lambda_i \geq 0, \mu_{ij} \geq 0, j \in C(i) \cup \{c(i)\}, i \in \mathcal{N}$, are the Lagrange multipliers.
An approach similar to Section~\ref{sec:aggr-interf-max-min-fair} prescribes
the following update rules. For all $i \in \mathcal{N}, j \in C(i) \cup \{c(i)\}$,
\begin{align}
p_i = &\begin{cases}
       \lambda_i \left(\sum_{j \in C(i)}\frac{\lambda_j}{1 + b_{ij} - p_i} \right)^{-1} & \mbox{if } \frac{1}{\lambda_i}\sum_{j \in C(i)}\frac{\lambda_j}{b_{ij}} > 1, \\
       1 & \mbox{otherwise},
      \end{cases}
\label{eqn:close-interf-p-update} \\
\theta_i = &\begin{cases}
       -\lambda_i -&\sum_{j \in C(i) \cup \{c(i)\}} (\mu_{ij} - \mu_{ji})  \\
                   &\mbox{if }  \lambda_i + \sum_{j \in C(i) \cup \{c(i)\}}(\mu_{ij} - \mu_{ji}) > 0, \\
       0 & \mbox{otherwise},
       \end{cases}
\label{eqn:close-interf-theta-update} \\
\lambda_i(n+1) = &\bigg[\lambda_i(n) + \beta(n) \bigg(\theta_i \nonumber \\
                 & \left. \left. - \log p_i - \log \left(1 - \frac{p_c(i)}{1 + b_{c(i)i}}\right) \right)\right]^+, \label{eqn:close-interf-lambda-update} \\
\mu_{ij}(n+1)  = &[\mu_{ij}(n) + \beta(n)(\theta_i - \theta_j)]^+, \label{eqn:close-interf-mu-update}
\end{align}
where $\beta(n)$ is the step size at the $n$th iteration as before.
Observe that any user $i$ can perform updates~\eqref{eqn:close-interf-p-update}-\eqref{eqn:close-interf-mu-update} via
local information exchange. More precisely, it needs to communicate only with user $c(i)$
and the users in $C(i)$.

\subsection{A Note on Distributed Implementation}
The aim of the present paper is primarily of theoretical
nature and it is beyond our scope
to discuss implementation issues. Let us however stress
that the discussed adaptations of the MAP are
implementable. We will focus on the
proportional fair case in view
the main focus of the paper.

Assume each receiver has a distinctive pilot signal
with fixed power. Since we assumed a quasi-static network in which the nodes move at a slower time scale, each
node can then learn the distance that separates
it from a given receiver by listening to
its pilot signal and by performing a time average over~(so as to smooth out fading).
Once this data is available for all receivers, a
given transmitter can then solve
the key fixed point equation that characterizes
its optimal MAP. Distinctive pilot signals can
be obtained by a collection of orthogonal codes
chosen at random by the receivers.
In practice, it is enough for a transmitter to detect
the ``dominant'' receivers~(i.e., those within
a certain distance to it), so that the scheme
will work in an infinite network
with a finite~(properly tuned) number of such codes.

\section{Stochastic Geometry Analysis}
\label{ss:stocg}

\subsection{Network and Communication Model}
We now assume that the transmitting nodes are scattered on
the Euclidian plane according to a homogeneous Poisson
point process of intensity $\lambda$. For each transmitter, its
corresponding receiver is at distance $r_0$ in a random direction.
The traffic and channel models are the same as
in Section~\ref{sec:network-model}.
As before the transmitters use slotted Aloha to access the channel,
and a receiver successfully receives the packet from its transmitter
if the received SINR exceeds a threshold $T$. Finally, transmitters adapt
their attempt probabilities as described in Section~\ref{sec:aggr-interf}.

Each transmitter is associated with
a multi dimensional mark that carries information about the adaptive transmission probability and
the transmission status.
Let $\tilde{\Phi} = \{X_n,Z_n\}$ denote a marked Poisson point, where
\begin{itemize}
\item $\Phi = \{X_n\}$ denotes the Poisson point process of
intensity $\lambda$, representing the location of transmitters in the
Euclidean plane.
\item $\{Z_n = (\phi_n, p_n, e_n)\}$ denote the marks of
the Poisson point process $\tilde{\Phi}$ , which consist of three
components:
\begin{itemize}
\item $\{\phi_n\}$ denote the angles from transmitters to receivers.
These angles are i.i.d. and uniform on $[0,2\pi]$
and independent of $\Phi$. We will call them the primary marks.

\item $\{p_n\}$ denote the MAPs of the nodes;
$p_n$ is a secondary mark (i.e.
a functionals of $\Phi$ and its primary marks,
see below).

\item $\{e_n\}$ are indicator functions that take value
one if a given node decides to transmit in a given
time slot, and zero otherwise.
Clearly, $\mathbb{P}(e_n = 1) = p_n = 1 - \mathbb{P}(e_n = 0)$.
In particular, given $p_n$, $e_n$ is
independent of everything else
including $\{e_m\}_{m \ne n}$.

\end{itemize}

\end{itemize}

The locations of the receivers will be denoted by
$\Phi^r = \{Y_n = X_n+(r_0,\phi_n)\}$
with $(r_0,\phi) := (r_0 \cos \phi,  r_0 \sin \phi)$.
It follows from the displacement theorem
\cite{stochproc-wireless.baccelli-blaszczyszyn09stochastic-geometry-wireless-networks-1} that $\Phi^r $ is also a homogeneous Poisson point
process of intensity $\lambda$.

The above assumptions will be referred to as the
Poisson model. We will also consider below
a more general case where the above marked
point process is just stationary.

\subsection{Proportional Fair Spatial Aloha}
\subsubsection{MAP distribution}
Let us consider response functions $L:\mathbb{R}^2 \times \mathbb{R}^2 \rightarrow \mathbb{R}^+ $
defined for each $0 \leq p \leq 1$ as follows
\[L_\rho(x,y) = \frac{\rho}{\frac{\Vert x- y\Vert^{\alpha}}{T r_0^{\alpha}} + 1 - \rho}.\]
For all $0 \leq \rho \leq 1, x \in \mathbb{R}^2$, the shot noise field
$J_{\Phi^r}(\rho,x)$ associated with the
above response function and the
marked point process $\tilde{\Phi}$ is
\begin{equation*}
J_{\Phi^r}(\rho,x) = \int_{\mathbb{R}^2}L_\rho(x,y)\Phi^r({\rm d} y)
              =\sum_{Y_n \in \Phi^r} L_\rho(x,Y_n).
\end{equation*}
Notice that this shot noise is not that representing the interference
at $x$.
It rather measures the effect of the presence of a transmitter
at $x$ on the whole set of receivers.

Consider a typical node at the origin, $X_0 = 0$,
with marks $p_0, \phi_0$. Let $\mathbb{P}^0$
denote the Palm distribution of the stationary marked point process
$\tilde{\Phi}$~\cite[Chapter~1]{stochproc-wireless.baccelli-blaszczyszyn09stochastic-geometry-wireless-networks-1}.
The fixed point equation determining the MAP of node $X_0=0$ reads~(see~\eqref{pfair-fixed-point})
\[ \frac{1}{p_0}=
\sum_{n \ne 0}
\frac{1}{\frac{\Vert Y_n\Vert^{\alpha}}{T r_0^{\alpha}} + 1 - p_0}.\]
We have a similar equation for each node and the
sequence $\{p_n\}$ is readily seen to be a sequence of
marks of $\Phi$ and its primary marks.

It follows directly from monotonicity arguments that
\[ \left\{\frac{1}{p_0}  <\frac 1 \rho\right\} \quad \mbox{iff}
\quad \left\{ \sum_{n \ne 0}
\frac{\rho}{\frac{\Vert Y_n\Vert^{\alpha}}{T r_0^{\alpha}} + 1 - \rho}
<1\right\}.\]
Notice that we have not used the specific assumptions on the
point process so far.
Hence we have the following general connection
between the optimal MAP distribution and the shot noise
$J_{\Phi^r}$:
\begin{theorem}
For all stationary marked point processes
$\widetilde \Phi$~(not necessarily Poisson),
for all $0 < \rho < 1$,
\begin{align*}
\mathbb{P}^0(p_0 > \rho) &= \mathbb{P}^0\left(J_{\Phi^r \setminus {\{Y_0\}}}(\rho,0) < 1\right), \\
\mbox{and} \ \ \ \ \ \ \ \ \ \ \ \ \ \ \ \ \ \ \ \ \ \ \ \ \ \ & \ \ \ \ \ \ \ \ \ \ \ \ \ \ \ \ \ \ \ \ \ \ \ \ \ \ \ \ \ \ \ \ \ \ \ \ \ \ \ \ \ \ \ \ \ \ \ \ \ \ \ \ \\
\mathbb{P}^0(p_0 = 1) &= \mathbb{P}^0\left(J_{\Phi^r \setminus {\{Y_0\}}}(1,0) < 1\right),
\end{align*}
with $\mathbb{P}^0$ the Palm distribution of $\tilde{\Phi}$.
\end{theorem}
We now use the fact that $\tilde{\Phi}$ is an independently
marked Poisson point
process~\cite[Definition~2.1]{stochproc-wireless.baccelli-blaszczyszyn09stochastic-geometry-wireless-networks-1}.
From Slivnyak's theorem~\cite[Theorem~1.13]{stochproc-wireless.baccelli-blaszczyszyn09stochastic-geometry-wireless-networks-1},
\[\mathbb{P}^0\left(J_{\Phi^r \setminus {\{Y_0\}}}(\rho,0) < 1\right) = \mathbb{P}\left(J_{\Phi^r}(\rho,0) < 1\right)\]
for all $0 \leq \rho \leq 1$. Consequently,
\begin{align*}
\mathbb{P}^0(p_0 > \rho) &= \mathbb{P}\left(J_{\Phi^r}(\rho,0) < 1\right), \\
\mbox{and} \ \ \ \ \ \ \ \ \ \ \ \ \ \ \ \ \ \ \ \ \ \ \ \ \ \ & \ \ \ \ \ \ \ \ \ \ \ \ \ \ \ \ \ \ \ \ \ \ \ \ \ \ \ \ \ \ \ \ \ \ \ \ \ \ \ \ \ \ \ \ \ \ \ \ \ \ \ \ \\
\mathbb{P}^0(p_0 = 1) &= \mathbb{P}\left(J_{\Phi^r}(1,0) < 1\right).
\end{align*}

It follows from~\cite[Proposition~2.6]{stochproc-wireless.baccelli-blaszczyszyn09stochastic-geometry-wireless-networks-1} and from the fact that
$\Phi^r$ is a homogeneous Poisson point process that
one can write the Laplace transform $\mathcal{L}_{J(\rho,0)}(s)$
of the shot noise $J_{\Phi^r}(\rho,0)$ as
\begin{equation}
\mathcal{L}_{J(\rho,0)}(s) = \exp \left\{-2 \pi \lambda \int_0^{\infty}\left(1 - \e^{-\frac{s\rho\bar{r}_0}{r^{\alpha} + (1 - \rho)\bar{r}_0}}\right)r {\rm d} r \right\},
\label{eq:lapl}
\end{equation}
where $\bar{r}_0 := T r_0^{\alpha}$.
\begin{theorem}
\label{thm:map-distribution}
Under the above Poisson assumptions,
the attempt probability of the typical node has the distribution
\begin{equation}
\mathbb{P}^0(p_0 > \rho) =\frac 1 {2\pi}
\int_{-\infty}^{\infty}\mathcal{L}_{J(\rho,0)}(i w)
\frac{\e^{iw} -1}{iw} {\rm d} w,
\label{eq:Parsev}
\end{equation}
with $\mathcal{L}_{J(\rho,0)}(\cdot)$ given by (\ref{eq:lapl}).
\end{theorem}
\begin{proof}
Let $g_{\rho}(\cdot)$ denote the density of the shot noise field
$J_{\Phi}(\rho,0)$. Then
\begin{equation*}
\mathbb{P}^0(p_0 > \rho) = \int_0^1 g_{\rho}(t) {\rm d} t
                         = \int_{-\infty}^{\infty} g_{\rho}(t) u(t)  {\rm d} t,
\end{equation*}
where $u(t) = 1$  if $ 0 \leq t \leq 1 $ and 0 otherwise.
Now using Parseval's theorem
\[
\mathbb{P}^0(p_0 > \rho) = \frac{1}{2\pi}\int_{-\infty}^{\infty} \mathcal{F}_{J(\rho,0)}(w) \mathcal{F}^*_{u}(w) {\rm d} w,
\]
with $\mathcal{F}_{A}(w)=\mathbb{E}\exp(-iwA)$ the Fourier
transform of the real valued random
variable $A$ and $B^*$ the complex conjugate of $B$.
The claim follows after substituting
$
\mathcal{F}_{u}(w) = \frac{1 - \e^{-iw}} {iw}
$
and $\mathcal{F}_{J(\rho,0)}(w) = \mathcal{L}_{J(\rho,0)}(iw)$.
\end{proof}

\begin{remark}
\label{LEM:LAPLACE-TRANSFORM}
For $\alpha = 4$, the Laplace transform $\mathcal{L}_{J(p,0)}(s)$ can be simplified as
\begin{align*}
\lefteqn{\mathcal{L}_{J(p,0)}(s)}  \\
 &= \exp \left\{-2 \pi \lambda \sqrt{(1-p)T}r_0^2 \int_0^1 \frac{1 - \e^{-spv^2/(1-p)}}{v^2 \sqrt{1 - v^2}} {\rm d} v \right\}.
 \end{align*}
\end{remark}

\subsubsection{Mean Utility}
This subsection is devoted to the analysis of the mean value of the
logarithm of the throughput of the typical node:
$$ \theta= \mathbb{E}^0\log((p_0)) + \mathbb{E}^0\log((q_0)),$$
with $q_0$ defined in~\eqref{eqn:succ-prob}.
Since we know the cdf $f$ of $p_0$, the first term poses no problem.
The second term can be rewritten as~(see~\eqref{eqn:succ-prob})
\begin{equation*}
\mathbb{E}^0\log((q_0)) =
\mathbb{E}^0 \left[\sum_{n \ne 0} \log \left(1-\frac{p_n}{\frac{\Vert X_n - Y_0\Vert^{\alpha}}{T r_0^{\alpha}} + 1}\right)
\right].
\end{equation*}
Under the law $\mathbb{P}^0$, the points $\{X_n\}_{n \neq 0}$ of $\Phi$
form a homogeneous Poisson point process of intensity $\lambda$.
However, the marks $\{p_n\}_{n \ne 0}$ do not have the law identified in the last
section. In fact, the mark $p_n$ of a point $X_n~(n \neq 0)$ satisfies
the following modified fixed point equation:
$$ \frac{1}{p_n}=
\frac{1}{\frac{\Vert X_n - Y_0 \Vert^{\alpha}}{T r_0^{\alpha}} + 1 - p_n}
+
\sum_{m \ne 0,n}
\frac{1}{\frac{\Vert X_n-Y_m\Vert^{\alpha}}{T r_0^{\alpha}} + 1 - p_n},$$
with the convention that $p_n=1$ if there is no solution in $[0,1]$.
We can use the same argument as above to conclude that
$ \frac{1}{p_n}  < \frac 1 \rho$ iff
$$
\frac{\rho}{\frac{\Vert X_n - Y_0\Vert^{\alpha}}{T r_0^{\alpha}} + 1 - \rho} +
\sum_{m \ne 0,n}
\frac{\rho}{\frac{\Vert X_n- Y_m\Vert^{\alpha}}{T r_0^{\alpha}} + 1 - \rho}
<1
.$$

Conditioned on there being two nodes at $0$ and $x$,
the other points form a homogeneous Poisson point process of intensity $\lambda$.
This allows one to prove the following.

\begin{theorem}
\label{thm:map-typical-receiver}
Under the above Poisson assumptions, given
that there is a node at 0 and a node at $x\in \mathbb{R}^2$,
the attempt probability of the node at $x$
has the distribution
\[\mathbb{P}^{0,x}(p_x > \rho) =\frac 1 {2\pi}
\int_{-\infty}^{\infty}\mathcal{L}_{J_{x}(\rho,0)}(i w)
\frac{\e^{iw} -1}{iw} {\rm d}w,\]
with
\begin{align*}
\lefteqn{\mathcal{L}_{J_{x}(\rho,0)}(s) = } \ \ \ \ \ \ \ \ \ \\
 & \frac{1}{2\pi}\int_0^{2\pi} \exp
\left(-
\frac{s\rho\bar{r}_0}{||x - (r_0,\phi)||^{\alpha} + (1 - \rho)\bar{r}_0}
\right) {\rm d} \phi \\
&  \exp \left\{-2 \pi \lambda \int_0^{\infty}\left(1 - \e^{-\frac{s\rho\bar{r}_0}{r^{\alpha} + (1 - \rho)\bar{r}_0}}\right)r {\rm d}r \right\},
\end{align*}
and $(r_0,\phi) := (r_0 \cos \phi,  r_0 \sin \phi)$.
\end{theorem}
Due to the circular symmetry, the first integral in the expression of $\mathcal{L}_{J_{x}(\rho,0)}(s)$
depends on $x$ only through $\Vert x \Vert$.
Thus the density of $p_x$ also depends on $\Vert x \Vert$ only,
and it will be
denoted by $f_r$ when $\Vert x \Vert = r$.
The density of $p_0$ identified in the last subsection
will be denoted by $f$.
The main result of this section is:
\begin{theorem}
\label{THM:MEAN-UTILITY}
Under the above Poisson assumptions, in the
proportional fair case, the mean utility of a typical node
is
\begin{eqnarray}
\theta & = & \int_0^1 \log(u)f({\rm d}u)\nonumber \\
&  +&
\frac 1 {2\pi}
\int\limits_{\phi \in (0,2\pi)}
\int\limits_{x\in \mathbb{R}^2}
\int\limits_v
\log\left(1-\frac{v\overline r_0}{||x-(r_0,\phi)||^\alpha+\overline r_0}\right)
\nonumber \\
&  & \hspace{2cm} {\rm d}\phi
f_{||x-(r_0,\phi)||}({\rm d}v) {\rm d}x.
\end{eqnarray}
\end{theorem}
\begin{proof}
See Appendix~\ref{proof:thm-mean-utility}.
\end{proof}

\subsection{Discussion of the Other Cases}
A preliminary concern when trying to use Euclidean
stochastic geometry of the infinite plane in
the maximum throughput and max-min fairness cases
is that it is not clear whether the associated infinite
dimensional optimization problems make sense
in the first place. In the proportional fair case,
each node computes its optimizing MAP in one step
as the solution of a fixed point equation that is
almost surely well defined~(in terms of a shot noise) even
in the infinite Poisson population case.
Unfortunately, this does not extend to the other two cases.

This does not mean that there is no hope at all.
Consider for instance the maximum throughput case accounting only for the
closest interferer,
and further simplify it by measuring interference
at the transmitter rather than at the receiver.
In other words, consider the same optimization
problem as in Section~\ref{sec:close-interf-max-thput} but with
$c(i)$ being the closest transmitter to transmitter $i$~(rather than to receiver $i$).
Then, in the Poisson case, the infinite dimensional optimization
problem can be shown to reduce to a {\em countable
collection of finite optimization problems}.
This follows from the fact that there are
no infinite ``descending chains'' in a homogeneous Poisson
point process~\cite[Chapter~2]{stochproc.franceschetti-meester08random-networks}~(a {\it descending chain}
is a sequence of nodes $i_1,i_2,\ldots$ such that
$c(i_n)=i_{n+1}$ for all $n$).
As a result, the Poisson point process
can be decomposed into a countable collection
of finite ``descending trees'', where each path from the leaves to the root
is a descending chain.
The associated optimization is hence well defined
and the problem can be reduced to evaluating
the solution of the optimization problem of
Section~\ref{sec:close-interf-max-thput} on the typical descending tree.
Hence, there is hope to progress on this and on related
cases. This will however not be pursued here and is left for future research.

\section{Numerical Results}
\label{ss:numerics}

In this section, we study the proposed adaptive spatial Aloha schemes quantitatively.
We compute various metrics formulated through the stochastic geometry based analysis,
and we also perform simulation. The simulation not only validates the analytical model, but also
illustrates the performance of schemes for which
we do not have an analytical characterization.

\subsection{Computation of the Integrals}

We used Maple and Matlab to evaluate the integrals of Section
\ref{ss:stocg}. The infinite integral that shows up in
the expression of the Laplace transform (\ref{eq:lapl}) is handled
without truncation by Maple and Matlab.
The singularity at $w=0$ in the contour integrals
(\ref{eq:Parsev}) leveraging Parseval's theorem
is a false singularity and it is also
handled without further work
by either Maple or Matlab.
The Matlab code is particularly efficient and is
used throughout the analytical evaluations described
below.

\subsection{Simulation Setting}
We consider a two dimensional square plane with side length $L$, and
$N$ nodes placed independently over the plane according to the uniform
distribution; this corresponds to $\lambda = N/L^2$ in the stochastic
geometry model.\footnote{A finite snapshot of a Poisson random proces
would contain a Poisson distributed number of nodes. However, for large $\lambda L^2$,
the Poisson random variable with mean
$\lambda L^2$ is highly concentrated around its mean. Thus we can use $\lambda L^2$ nodes
 for all the realizations in our simulation.} Each node has its receiver randomly located on the
unit circle around it, again as per the uniform
distribution. Thus $r_{ii} = 1$ for all $i$. We set $\alpha = 4$
and $T = 10$. To nullify the edge effect, we take into account
only the nodes falling in the $L/2 \times L/2$ square around the center while computing
various metrics. While all other parameters remain, we vary $L$ and $N$ for different simulations.
For each parameter set we calculate the average of the performance metric of interest over $1000$ independent
network realizations.

\subsection{Joint Validation of the Analysis and the Simulation}
\label{subsec:validation}

We validate the analytical expression against the simulation for the case
of proportional fair medium access. For illustration, we  plot the cumulative distribution function~(c.d.f.)
of the MAP in Figure~\ref{fig:map-cdf-analysis}. Here we set $L = 40$ and consider two values of $N$, $N = 400$ and
$N = 800$, which correspond to $\lambda = 0.5$ and $\lambda = 0.25$ respectively.
The plots show that the stochastic geometry based formula~(see Theorem~\ref{thm:map-distribution}) quite accurately predicts the nodes' behavior in simulation.

\begin{figure}[h]
\centering
\includegraphics[width=3.0in,height=2.5in]{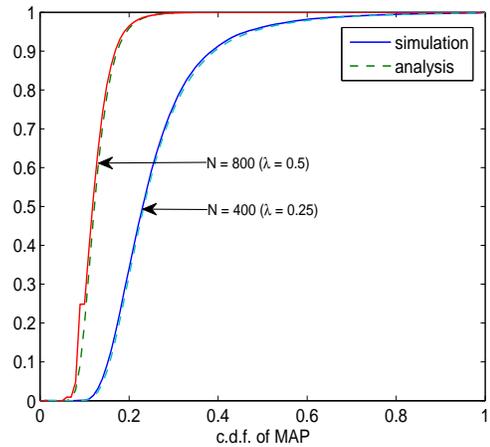}
\caption{Cumulative distribution function of the MAP for the proportional fair case.}
\label{fig:map-cdf-analysis}
\end{figure}

We also study the distributions $\mathbb{P}^{0,x}(p_x > \rho)$~(also referred to as $f_r$ for $\Vert x \Vert = r$)
defined in Theorem~\ref{thm:map-typical-receiver}.
Figure~\ref{fig:map-cdf-receiver} shows their plots for $\lambda = 0.25$ and two values of $\Vert x \Vert$,
$\Vert x \Vert = 1$ and $\Vert x \Vert = 10$.
Again the plots based on the analytical expression and those based on simulation closely match.
Clearly, under $\mathbb{P}^0$, nodes closer to the origin are more likely to be inactive.

\begin{figure}[h]
\centering
\includegraphics[width=3.0in,height=2.5in]{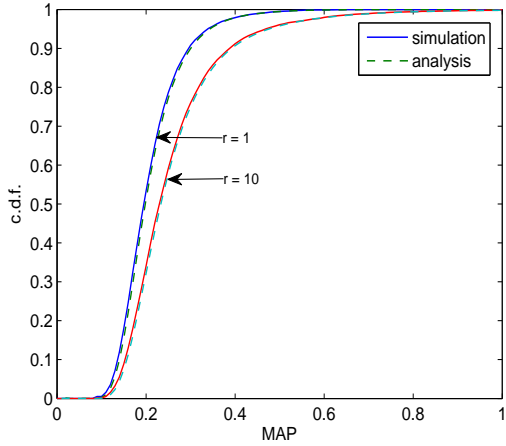}
\caption{Demonstration of $f_r$, the MAP distribution of a node at distance $r$ from origin, under $\mathbb{P}^0$.}
\label{fig:map-cdf-receiver}
\end{figure}

\subsection{Performance of the Adaptive Protocols}
\label{sec:performance}
In this section we illustrate the performance of
various adaptive schemes and their benefit over plain Aloha.
We compute the performance metrics via simulation and also through
analytical expressions whenever the latter are available. In such cases
the analytical results and the simulation validate each other.

First we set $L = 20$ and $N = 50$.
We consider the maximum throughput medium access, however, only
accounting for the closest interferers. In figure~\ref{fig:gibbs-closei}, we show steady state
behavior of our Gibbs sampling based algorithms; we have set the temperature
$\tau(t) = 1/\log(1 +t)$. As expected, the improved maximum throughput medium
access~(see the discussion at the end of Section~\ref{sec:close-interf-max-thput}) insures
a better exclusion behavior. Under this scheme, a lesser number of nodes transmit, and
neighboring nodes are unlikely to transmit simultaneously. So, this is expected to deliver better aggregate
throughput.

\begin{figure}[h]
\hspace{-.33in}
\includegraphics[width=4.0in,height=1.6in]{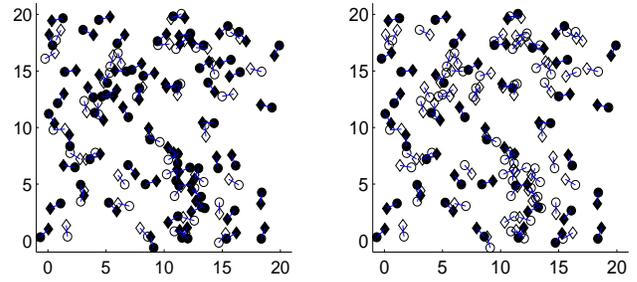}
\caption{Throughput maximizing medium access: There are 100 transmitter
receiver pairs. The {\it diamonds} represent
candidate transmitters, and the connected {\it circles} the corresponding receivers.
The solid diamonds represent the nodes that transmit in all the slots; others never
transmit. The left plot corresponds to the maximum throughput medium access and the right
one to its improved version.}
\label{fig:gibbs-closei}
\end{figure}

Now we keep $L$ fixed at $20$, but vary $N$ from $10$ to $100$; this corresponds to varying
$\lambda$ from $0.025$ to $0.25$ in the analytical expressions.
We evaluate the aggregate throughputs of various Aloha schemes including plain Aloha.
The average throughputs are plotted in Figure~\ref{fig:throughput}.
{\it Although in some of the schemes we derive the attempt probabilities only
considering the closest interferers, we always take into account the aggregate interference
while calculating the throughput}. When the number of nodes is small, both the throughput maximizing medium access
and plain aloha have identical performance; both prescribe attempt probabilities
close to one for all the nodes. When the number of nodes increases beyond $45$,
the throughput maximizing medium access significantly underestimates the interference,
and thus its performance deteriorates. On the other hand, the
aggregate interference based proportional fair scheme
significantly outperforms plain Aloha in terms of aggregate throughput also.
This benefit is sustained even as the number of nodes increases. We also notice that the improved
version of maximum throughput medium access~(see the discussion at the end of Section~\ref{sec:close-interf-max-thput})
yields best performance among all the schemes, and its
performance does not deteriorate until a much higher number of nodes.

\begin{figure}[h]
\centering
\includegraphics[width=3.0in,height=2.5in]{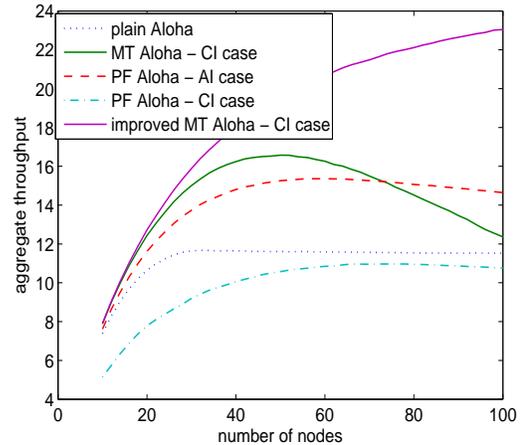}
\caption{Throughputs of various medium access schemes as a function of the number of nodes.
MT, PF, CI and AI stand for {\it maximum throughput, proportional fair, closest interferer}
and {\it aggregate interference} respectively.}
\label{fig:throughput}
\end{figure}

In Figure~\ref{fig:log-throughput}, we plot the logarithms of the aggregate throughputs
corresponding to the two proportional fair medium access schemes.
The figure illustrates that as the number of nodes
increases, the performance of the closest interferer based medium access
worsens in comparison to the performance of the aggregate interference based scheme -
this is not visible merely looking at the corresponding aggregate throughputs~(see Figure~\ref{fig:throughput})).

\begin{figure}[h]
\centering
\includegraphics[width=3.0in,height=2.5in]{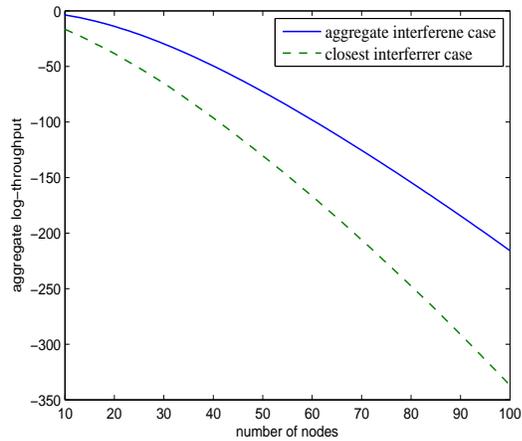}
\caption{Performance of the two proportional fair medium access schemes as a function
of the number of nodes.}
\label{fig:log-throughput}
\end{figure}

In Figure~\ref{fig:map-cdf}, we plot the c.d.f. of the MAP
for the proportional fair case. Our objective is to compare the case when nodes account for the
aggregate interference with when they account for the closest interferer only.
We set $L = 40$ and plot MAP distributions corresponding to two values of $N$, $400$
 and $800$~(corresponding to $\lambda = 0.25$ and $\lambda = 0.5$ respectively).
As expected, the nodes attempt more aggressively when they account for the closest interferer only.
While there is a favored probability interval in the aggregate interference case, there are more than one
such intervals in the closest interferer case. Also, in the latter case, about $34\%$ of nodes  attempt
in almost all slots, irrespective of $N$, the total number of nodes. This can be understood by
noticing that the probability that a node is not the closest interferer to any other node is not sensitive
to $N$.

\begin{figure}[h]
\centering
\includegraphics[width=3.0in,height=2.5in]{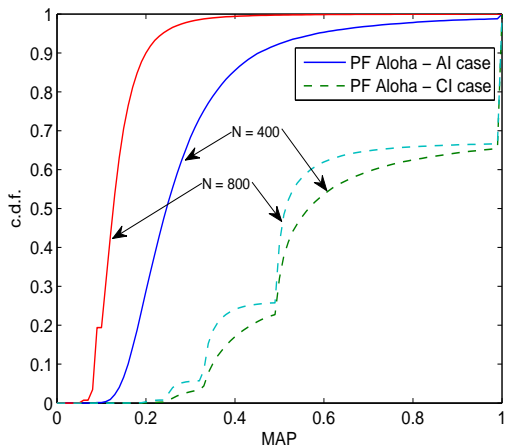}
\caption{Cumulative distribution function of the MAP for the proportional fair case.}
\label{fig:map-cdf}
\end{figure}

\section{Conclusion}
We have shown the feasibility of the performance analysis
of distributed adaptive protocols that aim at maximizing
some global utility in a large random network using stochastic
geometry.

More precisely, the most natural distributed adaptation of the medium
access probability of Aloha that aims at proportional fairness
optimization was shown to have a tractable optimal MAP distribution. This distribution
is obtained from the law of a certain shot noise field that describes
the interference created by a typical node to all receivers but his.
In the Poisson case, the distribution of the optimal MAP is obtained
as a non-singular contour integral which is amenable to an efficient
evaluation using classical numerical tools. The network performance
at optimum can in turn be deduced from the latter using Campbell's formula.

This approach is shown to provide an analytic way
of quantifying the gains brought by this proportionally
fair adaptive version of Aloha compared to plain Aloha.

This line of thoughts opens several research directions.
The first one is the extension to other types of fairness,
still in the framework of Aloha. We indicate that this is possible at least under
certain simplifications of the interference model.
The second and broader question is whether this approach
can be extended to MAC protocols other than Aloha.
An example would be an adaptation of the exclusion
radius of CSMA/CA to the full environment of a node aiming at
maximizing some utility of the throughput.
A third general question concerns the evaluation of the
``price of decentralization''.
When the discussed protocols are suboptimal because
of their greedy/distributed nature, is it possible
to use stochastic geometry to evaluate the typical discrepancy between the performance
of the distributed scheme and the optimal centralized one?

\section{Acknowledgments}
This work was carried out at Laboratory of Information, Networking
and Communication Sciences~(LINCS) Paris, and was supported by
INRIA-Alcatel Lucent Bell Labs Joint Research Center.

\remove{
%
\bibliographystyle{abbrv}
\bibliography{adaptive-spatial-aloha}  
%
%
}

\appendix
\section{Proof of Theorem~\ref{thm:mean-utility}}
\label{proof:thm-mean-utility}
Let $\Psi=\{Y'_m\}$ be a
point process with marks and $\psi\circ\theta_{z}$
be the point process $\psi$ shifted
by $-z$. Given that $Y_0=(r_0,\phi)$,
$$p_n=h(||X_n - (r_0,\phi)||,\Phi^r\setminus \{(r_0,\phi),Y_n\}\circ\theta_{X_n}),$$
where the mapping $h(u,\Psi)$ associates with
$\Psi$ and the real number $u$ the solution of
$$ \frac{1}{p}=
\frac{1}{\frac{u^{\alpha}}{T r_0^{\alpha}} + 1 - p}
+
\sum_{m}
\frac{1}{\frac{\Vert Y'_m\Vert^{\alpha}}{T r_0^{\alpha}} + 1 - p},$$
with the usual convention if there is no solution in $[0,1]$.
It follows from Slivnyak's theorem that
\begin{eqnarray*}
&&
\mathbb{E}^0
\left[\sum_{n\ne 0} \log
\left(1-\frac{p_n}{\frac{\Vert X_n - Y_0\Vert^{\alpha}}{T r_0^{\alpha}} + 1}\right)
\right]=\\
&&
\hspace{-.7cm}
\mathbb{E}^0
\left[\sum_{n\ne 0} \log
\left(1-\frac{h(||X_n - Y_0||,\Phi^r\setminus \{Y_0,Y_n\}\circ\theta_{X_n})}{\frac{\Vert X_n - Y_0\Vert^{\alpha}}{T r_0^{\alpha}} + 1}\right)
\right]=\\
&&
\hspace{-.7cm}\frac 1 {2\pi}
\int\limits_0^{2\pi}
\mathbb{E}
\sum_{n} \log
\left(1-\frac{h(||X_n - (r_0,\phi)||,\Phi^r \setminus \{Y_n\}\circ\theta_{X_n})}
{\frac{\Vert X_n-(r_0,\phi)\Vert^{\alpha}}{T r_0^{\alpha}} + 1}\right)
{\rm d}\phi.
\end{eqnarray*}
It now follows from Campbell's formula that
\begin{eqnarray*}
&& \mathbb{E}
\left[\sum_{n\ne 0} \log
\left(1-\frac{h(||X_n-(r_0,\phi)||,\Phi^r\setminus\{Y_n\}\circ\theta_{X_n})}{\frac{\Vert X_n-(r_0,\phi)\Vert^{\alpha}}{T r_0^{\alpha}} + 1}\right)
\right]\\
&&
\hspace{-.6cm}
=
\int_{\mathbb{R}^2}
\int_v  \log\left(1-\frac{v\overline r_0}{||x-(r_0,\phi)||^\alpha+\overline r_0}\right)
\lambda {\rm d}x\\ & & \hspace{3cm}
\mathbb{P}^0 ( h(||x-(r_0,\phi)||,\Phi^r\setminus\{Y_0\})= {\rm d}v)\\
&&
\hspace{-.6cm}
=
\int_{\mathbb{R}^2}
\int_v  \log\left(1-\frac{v\overline r_0}{||x-(r_0,\phi)||^\alpha+\overline r_0}\right))
\lambda {\rm d}x \\ & & \hspace{3cm} \mathbb{P} ( h(||x-(r_0,\phi)||,\Phi^r)= {\rm d}v),
\end{eqnarray*}
where the last relation follows from Slivnyak's theorem.

\remove{
\section{Proof of Remark~\ref{lem:laplace-transform}}
Recall that, for $\alpha = 4$,
\begin{equation}
\mathcal{L}_{J(p,0)}(s) = \exp \left\{-2 \pi \lambda \int_0^{\infty}\left(1 - e^{-\frac{sp\bar{r}_0}{r^4 + (1 - p)\bar{r}_0}}\right)r dr \right\} \label{eqn:lt-alpha-4}
\end{equation}
with $\bar{r}_0 := T r_0^4$. We denote the integral in the exponent with $Int$
In the following we simplify $Int$ via a series of substitutions.
\begin{align*}
Int &:=  \int_0^{\infty}\left(1 - e^{-\frac{sp\bar{r}_0}{r^4 + (1 - p)\bar{r}_0}}\right)r dr \\
    &\overset{\text{(i)}} = \int_0^{\infty}\left(1 - e^{-\frac{b}{t^2 + a}}\right) dt \\
    &= \int_0^{\infty}\left(1 - e^{-\frac{b/a}{(t/\sqrt{a})^2 + 1}}\right) dt \\
    &\overset{\text{(ii)}} = \sqrt{a}\int_0^{\infty}\left(1 - e^{-\frac{b/a}{u^2 + 1}}\right) du \\
    &\overset{\text{(iii)}} = \sqrt{a}\int_0^1 \frac{1 - e^{-(b/a)v^2}}{v^3} \frac{1}{\sqrt{1/v^2 - 1}} dv \\
    &= \sqrt{a}\int_0^1 \frac{1 - e^{-(b/a)v^2}}{v^2 \sqrt{1 - v^2}} dv \\
    &\overset{\text{(iv)}} = \sqrt{(1-p)T}r_0^2 \int_0^1 \frac{1 - e^{-spv^2/(1-p)}}{v^2 \sqrt{1 - v^2}}dv.
\end{align*}
In the above sequence of equalities,
\begin{enumerate}[(i)]
\item is obtained by setting $a = (1-p)\bar{r}_0, b = sp\bar{r}_0$ and $t = r^2$,
\item by setting $u = \frac{t}{\sqrt{a}}$,
\item by setting $v^2 = \frac{1}{1 + u^2}$, and
\item by substituting the values of $a,b$ and  $\bar{r}_0$,
\end{enumerate}
The desired result is obtained after substituting the final expression of $Int$ in~\eqref{eqn:lt-alpha-4}.
}

\end{document}